\documentclass{llncs}

\usepackage[usenames]{color}
\usepackage{llncsdoc,gastex,bbding,amssymb,amsmath,graphicx,amsmath,fullpage,verbatim,comment,amssymb,subfigure,algorithm,algorithmic}

\long\def\acks#1{\vskip 0.3in\noindent{\large\bf Acknowledgments}\vskip 0.2in
\noindent #1}

\newenvironment{reminder}[1]{\medskip
 
\noindent {\bf Reminder of Theorem #1.  }\em}{}
\newenvironment{reminderlemma}[1]{\medskip
 
\noindent {\bf Reminder of Lemma #1.  }\em}{}

\def \QED {\hfill{$\Box$}}

\newenvironment{proofof}[1]{\noindent {\em Proof of #1.  }}{\QED}

\def \N {{\mathbb N}}
\def \QED {\hfill{$\Box$}}

\def \poly {\textrm{poly}}

\begin{document}

\begin{frontmatter}

\title{Resolving the Complexity of Some Data Privacy Problems}
\author{Jeremiah Blocki\inst{1} \and Ryan Williams\inst{2}}
\institute{Carnegie Mellon University \email{jblocki@andrew.cmu.edu} \and IBM Almaden Research Canter \email{rrwilliams@gmail.com}}
\maketitle

\Asterisk An extended abstract of this work will appear in ICALP 2010.

\begin{abstract} We formally study two methods for data sanitation that have been used extensively in the database community: $k$-anonymity and $\ell$-diversity. We settle several open problems concerning the difficulty of applying these methods optimally, proving both positive and negative results:

\begin{itemize}
\item $2$-anonymity is in {\sf P}.
\item The problem of partitioning the edges of a triangle-free graph into 4-stars (degree-three vertices) is {\sf NP}-hard. This yields an alternative proof that $3$-anonymity is {\sf NP}-hard even when the database attributes are all {\em binary}. 
\item $3$-anonymity with only 27 attributes per record is {\sf MAX SNP}-hard.
\item For databases with $n$ rows, $k$-anonymity is in $O(4^n \cdot \poly(n)))$ time for all $k > 1$.
\item For databases with $\ell$ attributes, alphabet size $c$, and $n$ rows, k-Anonymity can be solved in $2^{O(k^2 (2c)^\ell)} + O(n \ell)$ time.  
\item $3$-diversity with binary attributes is {\sf NP}-hard, with one sensitive attribute.
\item $2$-diversity with binary attributes is {\sf NP}-hard, with three sensitive attributes.
\end{itemize}
\end{abstract}


\end{frontmatter}

\pagenumbering{arabic}

\section{Introduction}

The topic of {\em data sanitization} has received enormous attention in recent years. The high-level idea is to release a database to the public in such a manner that two conflicting goals are achieved: (1) the data is useful to benign researchers who want to study trends and identify patterns in the data, and (2) the data is not useful to malicious parties who wish to compromise the privacy of individuals. Many different models for data sanitization have been proposed in the literature, and they can be roughly divided into two kinds: {\em output perturbative} models (e.g., \cite{agrawalsrikant,differentialprivacy}) and {\em output abstraction} models (e.g., \cite{samarati2001pri,sweeney2002kam,machanavajjhala:dpb}). In perturbative models, some or all of the output data is perturbed in a way that no longer corresponds precisely to the input data (the perturbation is typically taken to be a random variable with nice properties). This include work which assumes {\em interaction} between the prospective data collector and the database, such as differential privacy. In abstraction models, some of the original data is suppressed or generalized ({\em e.g.} an age becomes an age range) in a way that preserves data integrity. The latter models are preferred in cases where data integrity is the highest priority, or when the data is simply non-numerical.

In this work, we formally study two data abstraction models from the literature, and determine which cases of the problems are efficiently solvable. We study $k$-anonymity and $\ell$-diversity.

\subsection{K-Anonymity}

The method of $k$-anonymization, introduced in \cite{samarati2001pri,sweeney2002kam}, is a popular method in the database community for publicly releasing part of a database while protecting individual identities in that database. Formally speaking, an instance of the $k$-anonymity problem is a matrix (a.k.a. database) with $n$ rows and $m$ columns with entries drawn from an underlying alphabet. Intuitively, the rows correspond to individuals and the columns correspond to various attributes of them. For hardness results, we study a special case called the {\em suppression model}, where the goal is to replace entries in the matrix with a special symbol $\star$ (called a `star'), until each row is identical to at least $k-1$ other rows. The intuition is that the information released does not explicitly identify any individual in the database, but rather identifies at worst a group of size $k$.
\footnote{This intuition can break down when combined with background knowledge~\cite{machanavajjhala:dpb}. However, our intent in this paper is not to critique the security/insecurity of these methods, but rather to understand their feasibility.}

A trivial way to $k$-anonymize a database is to suppress every entry (replacing all entries with $\star$), but this renders the database useless. In order to maximize the utility of the database, one would like to suppress the fewest entries---this is the {\sc $k$-Anonymity} problem with suppression. Meyerson and Williams~\cite{meyerson2004cok} proved that in the most general case, this is  a difficult task: {\sc $k$-Anonymity} is {\sf NP}-hard for $k \geq 3$, provided that the size of the alphabet is $\Omega(n)$. Aggarwal {\em et al.}~\cite{aggarwal2005at} improved this, showing that {\sc $3$-Anonymity} remains {\sf NP}-hard even when the alphabet size is 3.  Bonizzoni {\em et al.}~\cite{bonizzoni2007abt} further improved the result to show that  {\sc $3$-Anonymity} is {\sf APX}-hard, even with a binary alphabet. They also showed that {\sc $4$-Anonymity} with a constant number of attributes per record is {\sf NP}-hard. Two basic questions remain:

 \begin{enumerate} 
\item How difficult is $3$-anonymity with a small number of attributes per record? 
\item How difficult is the $2$-anonymity problem?
\end{enumerate} 

Addressing the two questions above, we discover both a positive and negative result.  On the positive side, in Section~\ref{2anonP} we present a polynomial time algorithm for {\sc $2$-Anonymity}, applying a result of Anshelevich and Karagiozova~\cite{anshelevich2007tbm}: 

\begin{theorem} \label{2anon}{\sc $2$-Anonymity} is in {\sf P}.\end{theorem}

The polynomial time algorithm works not only for the simple suppression model, but also for the most general version of $k$-anonymity, where for each attribute we are given a {\em generalization hierarchy} of possible ways to withhold data. 

In Section~\ref{constantattr}, we consider $k$-anonymity in databases where the number of attributes per record is constant. This setting seems to be the most relevant for practice: in a database of users, the number of attributes per user is often dwarfed by the number of users in the database. We find a surprisingly strong negative result.

\begin{theorem}\label{3anon27}
{\sc $3$-Anonymity} with just 27 attributes per record is {\sf MAX SNP}-hard. Therefore, {\sc $3$-Anonymity} does not have a polynomial time approximation scheme in this case, unless {\sf P $=$ NP}.
\end{theorem}

The proof uses an alphabet with $\Omega(n)$ cardinality. This motivates the question: how efficiently can we solve $k$-anonymity with a small alphabet and constant number of attributes per record? Here we can prove a positive result, showing that when the number of attributes is small and the alphabet is constant, there are subexponential algorithms for optimal $k$-anonymity for every $k > 1$.

\begin{theorem} \label{kAnonAlg}
For every $k > 1$, an optimal $k$-anonymity solution can be computed in  $O(4^n \poly(n))$ time, where $n$ is the total number of rows in the database.
\end{theorem}

\begin{theorem} \label{kAnonAlg2}
Let $\ell$ be the number of attributes in a database, let $c$ be the size of its alphabet, and let $n$ be the number of rows.  Then k-Anonymity can be solved in $2^{O(k^2 (2c)^\ell)} + O(n \ell)$ time.  
\end{theorem}

This improves on results in~\cite{chaytor2008fixed}. Theorem \ref{kAnonAlg2} implies that k-Anonymity is solvable in polynomial time whenever $\ell \leq (\log \log n)/{\log c}$ and $c \leq \log n$. Theorem \ref{kAnonAlg2} also implies that for $c = n^{o(1)}$ and $\ell = O(1)$, $k$-anonymity is solvable in {\em subexponential time}. Therefore it is highly unlikely that we can tighten the unbounded alphabet constraint of Theorem~\ref{3anon27}, for otherwise all of ${\sf NP}$ has $2^{n^{o(1)}}$ time algorithms. 

In Section~\ref{3anonbinary}, we provide an alternative proof that {\sc Binary $3$-Anonymity}, the special case of the problem where all of the attributes are binary-valued, is {\sc NP}-hard. This result is weaker than~\cite{bonizzoni2007abt} who recently showed that {\sc Binary $3$-Anonymity} is {\sc APX}-hard.  However, our proof also shows that a certain edge partitioning problem is {\sf NP}-complete, which to the best of our knowledge is new \footnote{{\sc Edge Partition Into Triangles} is NP-Complete as is {\sc Edge Partition Into 4-Stars} \cite{graphDecompIsNPC}, but this does not imply that {\sc Edge Partition Into Triangles and 4-Stars} is NP-Complete.}. Let {\sc Edge Partition Into Triangles and 4-Stars} be the problem of partitioning the edges of a given graph into 3-cliques (triangles) and 4-stars (graphs with three degree-1 nodes and one degree-3 node).

\begin{theorem}\label{edgepart} {\sc Edge Partition Into Triangles and 4-Stars}  is {\sc NP}-complete. \end{theorem}

Theorem~\ref{edgepart} implies that the {\sc Ternary $3$-Anonymity} hardness reduction given in~\cite{aggarwal2005at} is sufficient to conclude that {\sc Binary $3$-Anonymity} is {\sc NP}-hard.

\subsection{L-Diversity}

Finally, in Section~\ref{diversity} we consider the method of $\ell$-diversity introduced in~\cite{machanavajjhala:dpb}, which has also been well-studied. This method attempts to refine the notion of $k$-anonymity to protect against knowledge attacks on particular sensitive attributes.

We will work with a simplified definition of $\ell$-diversity that captures the essentials.  Similar to $k$-anonymity, we think of an $\ell$-diversity instance as a table (database) with $m$ rows (records) and $n$ columns (attributes). However, each attribute is also given a label $q$ or $s$, inducing a partition of the attributes into two sets $Q$ and $S$.  $Q$ is called the set of quasi-identifier attributes and $S$ called the set of sensitive attributes.

\begin{definition}
A database $D$ is said to be {\em $\ell$-diverse} if for every row $u_0$ of $D$ there are (at least) $\ell -1$ distinct rows $u_1,...,u_{\ell-1}$ of $D$ such that:
\begin{enumerate}
\item $\forall q \in Q, 0 \leq i < j < \ell$ we have $u_i[s] = u_j[s]$
\item $\forall s \in S, 0 \leq i < j < \ell$ we have $u_i[s]\neq u_j[s]$
\end{enumerate}
\end{definition}

Constraint 1 is essentially the same as $k$-anonymity.  Any row must have at least $k-1$ other rows whose (non sensitive) attributes are identical. Intuitively, Constraint 2 prevents anyone from definitively learning any row's sensitive attribute; in the worst case, an individual's attribute can be narrowed down to a set of at least $\ell$ choices.  Similar to $k$-anonymity with suppression, we allow stars to be introduced to achieve the two constraints.

 We can show rather strong hardness results for $\ell$-diversity.

\begin{theorem}\label{2div}
Optimal $2$-diversity with binary attributes and three sensitive attributes is {\sf NP}-hard.
\end{theorem}

\begin{theorem}\label{3div}
Optimal $3$-diversity with binary attributes and one sensitive attribute is {\sf NP}-hard.
\end{theorem}

Independent of their applications in databases, $k$-anonymity and $\ell$-diversity are also interesting from a theoretical viewpoint. They are natural combinatorial problems with a somewhat different character from other standard {\sf NP}-hard problems. They are a kind of discrete partition task that has not been studied much: {\em find a partition where each part is intended to ``blend in the crowd.''} Such problems will only become more relevant in the future, and we believe the generic techniques developed in this paper should be useful in further analyzing these new partition problems.

\section{Preliminaries}

We use $\poly(n)$ to denote a quantity that is polynomial in $n$.

\begin{definition} Let $n$ and $m$ be positive integers. Let $\Sigma$ be a finite set.
A \emph{database with $n$ rows (records) and $m$ columns (attributes)} is a matrix from $\Sigma^{n \times m}$. The \emph{alphabet} of the database is $\Sigma$. \end{definition}

\begin{definition}
Let $k$ be a positive integer. A database is said to be \emph{$k$-anonymous} or \emph{$k$-anonymized} if for every row $r_i$ there exist at least $k-1$ identical rows.
\end{definition}

As mentioned earlier, there are two methods of achieving $k$-anonymity: suppression and generalization.  In the suppression model, cells from the table are replaced with stars until the database is $k$-anonymous.  Informally, the generalization model allows the entry of an individual cell to be replaced by a broader category. For example, one may change a numerical entry to a range, {\em e.g.} (Age: 26 $\rightarrow$ Age: [20-30]). A formal definition is given in Section~\ref{generalization}.

In our hardness results, we consider $k$-anonymity with suppression. Since suppression is a special case of generalization, the hardness results also apply to $k$-anonymity with generalization.  Interestingly, our polynomial time $2$-anonymity algorithm works under both models.

\begin{definition}
Under the suppression model, the \emph{cost} of a $k$-anonymous solution to a database is the number of stars introduced.
\end{definition}

We let $Cost^{\star}_k(D)$ denote the minimum cost of $k$-anonymizing database $D$.

For the proofs of hardness, we introduce a few graph theoretical notions. Recall $K_n$ denotes the complete graph on $n$ vertices; we use the word {\em triangle} to denote $K_3$.

\begin{definition} Let $k \geq 1$. A {\em $k$-star} is a simple graph with $k-1$ edges, all of which are incident to a common vertex $v$.  $v$ is called the {\em center} of the $k$-star. The other $k-1$ vertices are called the {\em leaves} of the $k$-star.
\end{definition}

We also need a particular type of graph which we call a $3$-binary tree. All interior nodes of such a tree have degree three. 

\begin{definition}
Let $d\in \mathbb{N}$ be given. A \emph{$3$-binary tree of depth $d$} is a complete tree of depth $d$ where the root has three children and all other nodes have two children.
\end{definition}

For our inapproximability results, we need the notion of an L-reduction~\cite{papadimitriou1988optimization}.

\begin{definition}
Let $A$ and $B$ be two optimization problems and let $f:A\rightarrow B$ be a polynomial time computable transformation.  $f$ is an {\em L-reduction} if there are positive constants $\alpha$ and $\beta$ such that
\begin{enumerate}
\item $OPT(f(I)) \leq \alpha \cdot OPT(I)$
\item For every solution of $f(I)$ of cost $c_2$ we can in polynomial time find a solution of $I$ with cost $c_1$ such that
$$|OPT(I)-c_1|\leq \beta \cdot |OPT(f(I))-c_2|$$
\end{enumerate}
\end{definition}

\section{Polynomial Time Algorithm for 2-Anonymity}\label{2anonP}

Because 3-anonymity is hard even for binary attributes, it is natural to wonder if 2-anonymity is also difficult. However it turns out that achieving optimal 2-anonymity is polynomial time solvable. The resulting algorithm is nontrivial and would require heavy machinery to implement. We rely on a special case of hypergraph matching called {\sc Simplex Matching}, introduced in~\cite{anshelevich2007tbm}.

\begin{definition}
{\sc Simplex Matching}: Given a hypergraph $H = (V, E)$ with hyperedges of size 2 and 3 and a cost function $c : E \rightarrow \N$ such that
\begin{enumerate}
\item $(u,v,w) \in E(H) \Longrightarrow (u,v),(v,w),(u,w) \in E(H)$ and
\item $c(u,v)+c(v,w)+c(u,w) \leq 2 \cdot c(u,v,w)$
\end{enumerate}
find $M \subseteq E$ such that for all $v \in V$ there is a unique $e \in M$ containing $v$, and $\sum_{e \in M} c(e)$ is minimized.
\end{definition}

Anshelevich and Karagiozova gave a polynomial time algorithm to solve {\sc Simplex Matching}. We show that {\sc 2-Anonymity} can be efficiently reduced to a simplex matching.

\begin{reminder}{\ref{2anon}} {\sc $2$-Anonymity} is in {\sf P}.
\end{reminder}

\begin{proof}
Given a database $D$ with rows $r_1, \ldots, r_n$, let $C_{i,j}$ denote the number of stars needed to make rows $r_i$ and $r_j$.  Similarly define $C_{i,j,k}$ to be the number of stars needed to make $r_i, r_j, r_k$ all identical.  Observe that in a $2$-anonymization, any group with more than three identical rows could simply be split into subgroups of size two or three without increasing the anonymization cost. Therefore we may assume (without loss of generality) that the optimal $2$-anonymity solution partitions the rows into groups of size two or three.

Construct a hypergraph $H$ as follows:

\begin{enumerate}
\item For every row $r_i$ of $D$, add a vertex $v_i$.
\item For every pair $r_i,r_j$, add the 2-edge $\{v_i,v_j\}$ with cost $c(v_i,v_j) = C_{i,j}$.
\item For every triple $r_i,r_j,r_k$, add the 3-edge $\{v_i,v_j,v_k\}$ with cost $c(v_i,v_j,v_k) = C_{i,j,k}$.
\end{enumerate}

Thus $H$ is a hypergraph with $n$ vertices and $O(n^3)$ edges. We claim that $H$ meets the conditions of the simplex matching problem. The first condition is trivially met. Suppose we anonymize the pair of rows $r_i,r_j$ with cost $C_{i,j}$, so both rows have $\frac{1}{2} C_{i,j}$ stars when anonymized. Observe that if we decided to anonymize the group $r_i,r_j,r_k$, the number of stars introduced per row would not decrease. That is, for all $i,j,k$ we have \[\frac{1}{3} \cdot C_{i,j,k} \geq \frac{1}{2} \cdot C_{i,j}.\] By symmetry, we also have \[\frac{1}{3} \cdot C_{i,j,k} \geq \frac{1}{2} \cdot C_{j,k},~ \frac{1}{3} \cdot C_{i,j,k} \geq \frac{1}{2} \cdot C_{i,k}.\] Adding the three inequalities together,  $$ C_{i,j,k} \geq \frac{1}{2} (C_{i,j} + C_{j,k}+C_{i,k}).$$ Therefore $H$ is an instance of the simplex matching problem.

Finally, observe that any simplex matching of $H$ corresponds to a $2$-anonymization of $D$ with the same cost, and vice-versa.  \QED \end{proof}

\subsection{The general case}\label{generalization}

The proof of Theorem~\ref{2anon} also carries over to the most general case of {\sc $2$-Anonymity}, where instead of only suppressing entries with stars, we have a {\em generalization hierarchy} of possible values to write to an entry. We give a simple definition of generalization hierarchy that captures the essential features described in~\cite{sweeney2002kam}.

\begin{definition} Let $\Sigma$ be an alphabet of attributes, and let $\Gamma \supsetneq \Sigma$.
A {\em generalization hierarchy} is a rooted tree $T$ on $|\Gamma|$ nodes with $|\Sigma|$ leaves, where the vertices $v$ in $T$ are put in one-to-one correspondence with alphabet symbols $a(v) \in \Gamma$, the leaves are in one-to-one correspondence with the symbols of $\Sigma$, and all vertices $v$ have a cost $c(v) \in \mathbb{N}$. The cost function satisfies the property that if $u$ is the parent of $v$ in $T$, then $c(u) \geq c(v)$.  \end{definition}

The key property of a generalization hierarchy is that the cost function decreases as one moves from the root of $T$ down to its leaves. Note the suppression model of {\sc $k$-Anonymity} can be modeled with a trivial generalization hierarchy: we can take a star graph $T$ where the center of the star has symbol $\star$ and cost $1$, while the leaves (corresponding to the letters of $\Sigma$) have cost $0$. For any generalization hierarchy $T$, one can define the {\sc $k$-Anonymity-$T$} problem, where the goal is to replace some entries in a matrix from $\Sigma^{n \times m}$ with symbols in $\Gamma-\Sigma$ such that (a) every row is identical to at least $k-1$ other rows, (b) every symbol replaced is a {\em successor} of the new alphabet symbol replacing it (in $T$), and (c) the total sum of costs associated with these new symbols is minimized.

\begin{theorem} For every generalization hierarchy $T$,  {\sc $k$-Anonymity-$T$} is in {\sf P}.
\end{theorem}

\begin{proof} (Sketch) We define a hypergraph $H$ just as in Theorem~\ref{2anon}, but with new costs $C_{i,j}$ and $C_{i,j,k}$ reflecting the costs of a particular generalization hierarchy. One can still prove that the conditions for the simplex matching problem hold, using the fact that if $u$ is any ancestor of $v$, then $c(u) \geq c(v)$. This condition implies that if we have anonymized two rows $r_i, r_j$, adding a third row $r_k$ to be anonymized cannot {\em decrease} the cost of anonymization per row. That is, the particular generalization symbols needed to make all three rows identical could only cost more than the symbols needed to anonymize the two rows originally.  \QED \end{proof}

\section{$k$-Anonymity With Few Attributes}
\label{constantattr}

We now turn to studying the complexity of $k$-Anonymity with a constant number of attributes. First we show that for an unbounded alphabet, the $3$-anonymity problem is still hard even with only 27 attributes. We use the following {\sf MAX SNP}-hard problem in our proof.

\begin{definition}
{\sc Max 3DM-3} {\em (Maximum 3-Dimensional Matching With 3 Occurrences)}\\
\noindent {\sc Instance:} A set $M \subseteq W\times X \times Y$ of ordered triples where $W,X$ and $Y$ are disjoint sets.  The number of occurrences in $M$ of any element in $W,X$ or $Y$ is bounded by 3. Let $$C_{3DM}(M') = \frac{3|M'|}{{|W|+|X|+|Y|}}.$$
\noindent {\sc Goal:} Maximize $C_{3DM}(M')$ over all $M' \subseteq M$ such that no two elements of $M'$ agree in any coordinate.
\end{definition}

\begin{reminder}{\ref{3anon27}}
{\sc $3$-Anonymity} with just 27 attributes per record is {\sf MAX SNP}-hard. Therefore, {\sc $3$-Anonymity} does not have a polynomial time approximation scheme in this case, unless {\sf P $=$ NP}.
\end{reminder}

\begin{proof} To show 3-anonymity is {\sf MAX SNP}-hard, we show that there is an L-reduction from {\sc Max 3DM-3} to {\sc 3-Anonymity with 27 Attributes}~\cite{papadimitriou1988optimization}, since it is known that {\sc Max 3DM-3} is {\sf MAX SNP}-complete ~\cite{kann1991mbd}. 

Given a {\sc Max 3DM-3} instance $I = (M, W, X,Y)$, construct a {\sc 3-Anonymity} instance $D$ as follows:
\begin{enumerate}
\item Define $\Sigma = M \bigcup W \bigcup X \bigcup Y$, so that it contains a special symbol for each triple in $t \in M$ and each element $r \in  W \bigcup X \bigcup Y$.
\item Add a row to $D$ corresponding to each element $r_i \in  W \bigcup X \bigcup Y$, as follows. For $r \in W\bigcup X \bigcup Y$, let $t_{r,1}, t_{r,2},t_{r,3}\in M$ be the three triples of $M$ which contain $r$  (if there are less than three triples then simply introduce new symbols).
\smallskip
\begin{itemize}
\item  If $r \in W$ then add the following row to $D$:
\begin{center}
\begin{tabular}{|c|c|c|c|c|c|c|c|c|c|c|c|c|}
\hline
 $t_{r,1}$ & $t_{r,1}$ & $t_{r,1}$ & $t_{r,1}$ & $t_{r,1}$ & $t_{r,1}$ & $t_{r,1}$ & $t_{r,1}$ & $t_{r,1}$ & $t_{r,2}$ & $t_{r,2}$ & $\ldots$ & $t_{r,3}$\\
\hline
\end{tabular}
\end{center}
\smallskip
That is, the row contains nine copies of $t_{r,1}$, nine copies of $t_{r,2}$, then nine copies of $t_{r,3}$.

\item If $r \in X$, then add the row:
\begin{center}
\begin{tabular}{|c|c|c|c|c|c|c|c|c|c|c|c|c|}
\hline
$t_{r,1}$ & $t_{r,1}$ & $t_{r,1}$ & $t_{r,2}$ & $t_{r,2}$ & $t_{r,2}$ & $t_{r,3}$ & $t_{r,3}$ & $t_{r,3}$ & $t_{r,1}$ & $t_{r,1}$ & $\ldots$ & $t_{r,3}$\\
\hline
\end{tabular}
\end{center}

\item If $r \in Y$, then add the row:
\begin{center}
\begin{tabular}{|c|c|c|c|c|c|c|c|c|c|c|c|c|}
\hline
 $t_{r,1}$ & $t_{r,2}$ & $t_{r,3}$ & $t_{r,1}$ & $t_{r,2}$ & $t_{r,3}$ & $t_{r,1}$ & $t_{r,2}$ & $t_{r,3}$ & $t_{r,1}$ &$t_{r,2} $ & $\ldots$ & $t_{r,3}$ \\
\hline
\end{tabular}
\end{center}
\end{itemize}
\end{enumerate}

Suppose $w_i \in W, x_j \in X,y_k \in Y$ are arbitrary. Then the corresponding three rows in the database have the form:
\begin{center}

\begin{tabular}{|c||c|c|c|c|c|c|c|c|c|c|c|}
\hline
$w_i$ & $t_{w_i,1}$ & $t_{w_i,1}$ & $t_{w_i,1}$ & $t_{w_i,1}$ & $t_{w_i,1}$ & $t_{w_i,1}$ & $t_{w_i,1}$ & $t_{w_i,1}$ & $t_{w_i,1}$ & $\ldots$ & $t_{w_i,3}$\\
$x_j$ & $t_{x_j,1}$ & $t_{x_j,1}$ & $t_{x_j,1}$ & $t_{x_j,2}$ & $t_{x_j,2}$ & $t_{x_j,2}$ & $t_{x_j,3}$ & $t_{x_j,3}$ & $t_{x_j,3}$ & $\ldots$ & $t_{x_j,3}$\\
$y_k$ & $t_{y_k,1}$ & $t_{y_k,2}$ & $t_{y_k,3}$ & $t_{y_k,1}$ & $t_{y_k,2}$ & $t_{y_k,3}$ & $t_{y_k,1}$ & $t_{y_k,2}$ & $t_{y_k,3}$ & $\ldots$ & $t_{y_k,3}$ \\
\hline
\end{tabular}
\end{center}

Observe that $D$ has a total of $27n$ entries, where $n = |X|+|W|+|Y|$.  Recall $Cost^{\star}_3(D)$ is the optimal number of stars needed to $3$-anonymize $D$. It is useful to redefine $3$-Anonymity as maximization problem (where one maximizes the information released). Let $P$ be a $3$-anonymous solution to $D$, and define
$$C_{3ANON}(P) = 1-\frac{Cost^{\star}_3(P)}{27 n},$$ so that $OPT(D) = \max_P \{C_{3ANON}(P)\}$.

Suppose $D = \{r_1,\ldots,r_n\}$ is an instance of {\sc 3-Anonymity} obtained from the above reduction. Three properties are immediate from the construction of $D$: 
\begin{enumerate}
\item For any $x$ rows  $r_i,r_j,r_k,r_l$, where $x \geq 4$, the cost of anonymizing these rows is
$$C_{i,j,k,l} = 27x$$ because there is no alphabet symbol that is used in all $4$ rows.
\item If \mbox{$\{r_i,r_j,r_k\} \notin M$} then the cost of anonymizing the three corresponding rows is
$$C_{i,j,k} = 3\cdot 27 = 81$$ because there is no alphabet symbol that is used in all $3$ rows.
\item If \mbox{$\{r_i,r_j,r_k\} \in M$} then the cost of anonymizing the three corresponding rows is
$$C_{i,j,k} = 3\cdot 26 = 78$$
because the three rows match in exactly one of the 27 columns.
\end{enumerate}

These properties lead directly to the lemma:

\begin{lemma} \label{kAnonConstAttribReductionLemma}
There is a polynomial time mapping $g$ from 3DM-3 feasible solutions to 3-anonymity feasible solutions, such that if $M' \subseteq M$ is a 3DM-3 feasible solution then $C_{3DM}(M') = 27 C_{3ANON}(g(M'))$.
\end{lemma}

The proof of Lemma~\ref{kAnonConstAttribReductionLemma} is given in Appendix~\ref{kAnonConstAttribReductionLemmaProof}.

It remains for us to show that the above reduction is in fact an $L$-reduction.  Let $I$ be a {\sc Max 3DM-3} instance,  with corresponding 3-Anonymity instance $f(I)$, and set $\alpha = \frac{1}{27}, \beta = 27$.  Now by Lemma~\ref{kAnonConstAttribReductionLemma}
$$OPT(f(I)) = \frac{1}{27} OPT(I) \leq \alpha OPT(I)$$
so that condition (1) of an L-reduction holds.  Similarly, if we have a solution of $f(I)$ of cost $c_2$, then again by Lemma~\ref{kAnonConstAttribReductionLemma} we can quickly compute a solution to $I$ of cost $c_1 = 27 c_2$.  Therefore,
$$|OPT(I)-c_1| = |27  OPT(f(I))-27 c_2| = \beta|OPT(f(I))-c_2|$$
so that condition (2) also holds. \QED
\end{proof}

\smallskip

To complement the above bad news, we now give an efficient algorithm for optimal $k$-anonymity when the number of attributes and the size of the alphabet are both small. Along the way, we also give an algorithm for the general $k$-anonymity problem that runs in roughly $4^n$ time (again $n$ is the number of rows). 

A naive algorithm for $k$-anonymity would take an exorbitant amount of time, trying all possible partitions of $n$ rows into groups with cardinality between $k$ and $2k-1$. We can reduce this greatly using a divide-and-conquer recursion.

\begin{reminder}{\ref{kAnonAlg}} 
For every $k > 1$, an optimal $k$-anonymity solution can be computed in $O(4^n \poly(n))$ time, where $n$ is the total number of rows in the database.
\end{reminder}

\begin{proof}  Interpret our $k$-anonymity instance $S$ as a multiset of $n$ vectors drawn from $|\Sigma|^{\ell}$.  Define $S_k = \{T: T\subseteq S, |T| \in [\frac{n}{2},\frac{n}{2}+2k]\}$. That is, $S_k$ contains all multisubsets which have approximately $n/2$ elements. Then \begin{equation}\label{kanonalg} Cost_k(S) = argmin_{T \in S_k} \left[Cost_k(S-T) + Cost_k(T)\right],\end{equation}  where $Cost_k(S)$ is the cost of the optimal $k$-anonymous solution for $S$. Equation \eqref{kanonalg} holds because (without loss of generality) any $k$-anonymized group of rows in a database is at most $2k-1$, so we can always partition the $k$-anonymized groups of a database into two multisets where their cardinalities are in the interval $[n/2-2k,n/2+2k]$. 

Suppose we compute the optimal $k$-anonymity solution by evaluating equation \eqref{kanonalg} recursively, for all eligible multisubsets $T$. In the base case when $|S| \in [k,2k-1]$, we make all rows in $S$ identical and return that solution.
 
We can simply enumerate all $2^n$ multisubsets of $S$ to produce all possible $T$ in equation~\eqref{kanonalg}. The time recurrence of the resulting algorithm is 
\[T(n) \leq 2^{n+1} \cdot T(n/2 + 2k) + 2^n.\]
This recurrence solves to $T(n) \leq O(\log n \times 2^{\log n} 2^{n + \frac{n}{2}+\frac{n}{4}+\cdots + 1}) \leq O(4^n \cdot \poly(n))$ for constant $k$. \QED
\end{proof}

\begin{reminder}{\ref{kAnonAlg2}}
Let $\ell$ be the number of attributes in a database, let $c$ be the size of its alphabet, and let $n$ be the number of rows.  Then k-Anonymity can be solved in $2^{O(k^2 (2c)^{\ell})} + O(n\ell)$ time.  
\end{reminder}

If $\ell$ and $c$ are constants then there are at most $c^\ell$ possible rows.  To specify a group of k-anonymized rows we write $G = <r',t>$ where $t$  is the number of times the anonymized row $r'$ occurs in the group.  We can think of a k-anonymous solution as a partition of the rows into such anonymized groups. The following lemma will be useful for our algorithm.
 
\begin{lemma} \label{constAttributeAndAlphabetLemma}
Suppose that our database $D$ contained at least $k(2k-1)\times 2^\ell$  copies of a row $r$.  Then the optimal $k$-anonymity solution must contain a group containing just row $r$, ie. $G = <r,t>$ where $t \geq k$.
\end{lemma}
\begin{proof}
Suppose for contradiction that our database contains more than $ k(2k-1)\times 2^\ell$ copies of row $r$, but that our optimal solution did not contain a group $G = <r,t>$.  Without loss of generality we can assume $k \leq t \leq 2k-1$ for each group since larger groups could be divided into two groups without increasing the cost. Therefore, we must have at least $k \times 2^\ell$ groups $G = <r',t>$ which contain the row $r$.  Notice that each attribute of $r'$ either matches $r$ or is a $\star$.  Hence, there are at most $2^\ell$ possible values of $r'$. By the pigeonhole principle there must be at least $k$ groups $G_i = <r',t_i>$ containing $r$ whose anonymized rows $r'$ are all identical.  Merge these groups into one big group \mbox{$G = <r', \Sigma_{i=1}^k t_i>$} at no extra cost.  Each of the original $k$ groups contained at least one copy of the row $r$ so we can split $G$ into two groups: $<r,k>$ and \mbox{$G' = <r',\Sigma_{i=1}^k t_i - k>$} while saving at least $k$ stars.  Hence, our original solution was not optimal.  Contradiction! \QED
\end{proof}

For each row $r$ we can define $Index(r)$ to be a unique index between $0$ and $c^l-1$ by interpreting $r$ as a $\ell$ digit number base $c$. Using Lemma \ref{constAttributeAndAlphabetLemma}, the following algorithm can be used to obtain a {\em kernelization} of a instance of k-Anonymity, in the sense of parameterized complexity~\cite{flum2006parameterized}. 

\begin{algorithm}
\caption{k-anonymize a database D with small alphabet and few attributes}
\label{constAttributeAndAlphabetAlgorithm}
\begin{algorithmic}
\REQUIRE $rowCount[i] = \|\{r\in D|Index(r) = i\} \|$
\REQUIRE $c, \ell$ small
\STATE $T \leftarrow  k(2k)\times 2^\ell$ 
\FOR{Row $r \in D$}
\STATE $i \leftarrow Index(r)$
\IF{$rowCount[i] > T $ }
 \PRINT "$<r, k>$"  \COMMENT{By Lemma \ref{constAttributeAndAlphabetLemma}}
 \STATE $rowCount[i] \leftarrow rowCount[i]-k$
\ENDIF
\ENDFOR

\COMMENT{Now $\forall i, rowCount[i] \leq T$ so there are at most $m \leq k(2k-1)2^\ell(c^\ell) = k(2k-1)(2c)^\ell$ rows remaining}

\end{algorithmic}
\end{algorithm}

\begin{lemma} \label{constAttributeAndAlphabetLemma2}
Algorithm \ref{constAttributeAndAlphabetAlgorithm} runs in $O(n \ell)$ time on a database $D$ and outputs a database $D'$ with at most $O(k^2(2c)^\ell)$ rows, with the property that an optimal k-anonymization for $D'$ can be extended to an optimal k-anonymization for $D$ in $O(n \ell)$ time.
\end{lemma}
That is, for the parameter $k+c+\ell$, the k-anonymity problem is not only fixed parameter tractable, but can also be efficiently kernelized.
\begin{proof} (Sketch) By implementing rowCount as a hash table, each $Index(r)$ and lookup opperation takes $O(\ell)$ time.  Hence, set up takes $O(n \ell)$ time as does the loop.  By lemma \ref{constAttributeAndAlphabetLemma} there must be an optimal k-anonymity solution containing $<r,t>$ with $t \geq k$ whenever $r$ occurs at least $k(2k-1)2^\ell$ times in $D'$.  Therefore, if $r$ occurs more than $k(2k)2^\ell > k+k(2k-1)2^\ell$ times in $D$ then there is an optimal k-anonymity solution which contains the groups $<r,k>$ and $<r,t>$ so adding back $k$ copies of row $r$ to $D'$ does not change the optimal k-anonymization except for the extra $<r,k>$ group.   \QED
\end{proof}

\begin{proofof}{Theorem~\ref{kAnonAlg2}}
By lemma \ref{constAttributeAndAlphabetLemma2}, Algorithm \ref{constAttributeAndAlphabetAlgorithm} takes an arbitrary $k$-anonymity instance $D$ and reduces it to a new instance $D'$ with at most $ k(2k)(2c)^\ell$ rows in time $O(n)$.  We can then apply Theorem \ref{kAnonAlg} to k-anonymize $D'$ in time $O(4^m \poly(m))$. The total running time is $2^{O(k^2 (2c)^\ell)} + O(n\ell)$.
\end{proofof}

\section{Hardness of 3-Anonymity With Binary Attributes}
\label{3anonbinary}
In 2005, Aggarwal {\em et al.}~\cite{aggarwal2005at} showed that $3$-anonymity with a ternary alphabet is {\sf NP}-hard. Their proof of hardness gives a reduction from {\sc Edge Partition Into Triangles}, in which one is given a graph and is asked to determine if the edge set $E$ can be partitioned into $3$-sets such that each set corresponds to a copy of $K_3$. In particular, Aggarwal {\em et al.} first present a reduction from the problem of  {\sc Edge Partition Into Triangles And 4-Stars}\footnote{This problem is: Given a graph $G = (V,E)$, is it possible to partition the edge set $E$ into $3$-sets such that each $3$-set corresponds to either a copy of $K_3$ or a $4$-star?} into {\sc Binary 3-Anonymity}. Then they introduce a third alphabet symbol to distinguish 4-stars from triangles in the reduction, concluding that a   {\sc Ternary 3-Anonymity} algorithm can be used to solve {\sc Edge Partition Into Triangles}. 

We shall strengthen this result by directly proving that the {\sc Edge Partition Into Triangles And 4-Stars} problem is  {\sf NP}-Complete. In fact, we establish the hardness of edge partitioning into 4-stars on {\em triangle-free} graphs. Using the aforementioned reduction of Aggarwal {\em et al.}, the hardness of {\sc Binary 3-Anonymity} follows from this result.

\begin{reminder}{\ref{edgepart}} {\sc Edge Partition Into 4-Stars} is {\sf NP}-Complete, even for triangle-free graphs. \end{reminder}

We describe the setup for Theorem~\ref{edgepart} in the following paragraphs. The reduction will be from {\sc 1-in-3 Sat}, which is well-known to be {\sf NP}-Complete~\cite{schaefer1978csp}. Recall that in the {\sc 1-in-3 Sat} problem, we are given a 3-CNF formula $\phi$ and are asked if there is a satisfying assignment to $\phi$ with the property that \emph{exactly} one literal in each clause is true. We call a yes-instance of the problem {\em 1-in-3 satisfiable}. Given a formula $\phi$, the idea of our reduction is to create triangle-free graph gadgets-- a gadget for each variable, and another type of gadget for each clause-- and connect them in a (triangle-free) way such that $\phi$ is 1-in-3 satisfiable if and only if the resulting graph can be edge-partitioned into 4-stars. We first define a type of graph that shall be used to simulate the truth assignment of a variable in $\phi$.

\begin{definition} Let $d \in \mathbb{N}$ be given. The graph \emph{$G_d$} is formed by taking two 3-Binary trees of depth $d$, deleting a leaf from exactly three different parents in each tree, and adding three edges so that the parents of deleted leaves in one tree are matched with the parents of deleted leaves in the other tree.
\end{definition}

In a copy of $G_d$, we consider all edges adjacent to leaves to be \emph{shared edges}, while all other edges are considered \emph{private}. Intuitively, the shared edges are those that are {\em shared} with other gadgets in our final graph. We say that $G$ contains a {\em share-respecting copy of $G_d$} if its vertex set can be partitioned into two sets $S$ and $T$ such that $S$ is an induced copy of $G_d$, and all edges crossing the cut $(S,T)$ are adjacent to shared edges in $S$.

To distinguish between the two trees in a copy of $G_d$, they are arbitrarily designated as the {\em top tree} and {\em bottom tree}, respectively.

The key property of the gadget $G_d$ is given by the following claim, which says that (in a certain sense) the edges of $G_d$ can be partitioned into $4$-stars in precisely two ways.  Figure~\ref{variableGadget} illustrates $S_5$, where the dashed edges are shared and the solid edge is private.  It can be found in Appendix~\ref{edgePartInto4StarsExamples}.

\begin{lemma}  \label{truefalse} Let $G$ be a graph containing a share-respecting copy of $G_d$. Assuming there is an edge partition of $G$ into $4$-stars, exactly one of two cases must hold for that partition:  
\begin{enumerate}
\item All shared edges belonging to the top tree of $G_d$ are contained in $4$-stars with centers in $G_d$, while all shared edges belonging to the bottom tree are contained in 4-stars with centers not contained in $G_d$.
\item All shared edges belonging to the bottom tree of $G_d$ are contained in 4-stars with centers in $G_d$, while all shared edges belonging to the bottom tree are contained in 4-stars with centers not contained in $G_d$.
\end{enumerate}
\end{lemma}

In the first case of the claim, we say that the copy of $G_d$ is {\em true partitioned}, and in the second case we say that $G_d$ is {\em false partitioned}.  Intuitively, each copy of $G_d$ in our final graph will correspond to a variable in $\phi$, and a {\em true/false} partition shall correspond to assigning that variable {\em true/false}.  Lemma \ref{truefalse} is proved in Appendix~\ref{proofOfGadgetClaim}.

We now define another type of graph that shall be used as gadgets to represent clauses in a given 1-in-3 SAT formula.

\begin{definition} The graph $S_5$ is a 5-star with one of its edges labeled \emph{private} and the other three edges labeled \emph{shared}.\end{definition}

Figure~\ref{4StarClauseGadget} illustrates $S_5$, where the dashed edges are shared and the solid edge is private.  It can be found in Appendix~\ref{edgePartInto4StarsExamples}.

Suppose a graph $G$ contains a share-respecting copy of $S_5$\footnote{A copy of $S_5$ is \em{share respecting} if and only the center vertex has degree $4$ and the leaf incident to the private edge has degree $1$.}, so that one node adjacent to the private edge of $S_5$ has degree one. Call this node $v$ and its adjacent node $u$ (the center of $S_5$). Then, any partition of $G$ into $4$-stars must contain a $4$-star with $u$ as its center, using the edge $(u,v)$. But this $4$-star must use two of the shared edges in $S_5$. Therefore an edge-partition of $G$ into $4$-stars is possible if and only if exactly one of the shared edges in $S_5$ participates in a $4$-star with a center that is outside of $S_5$.

We are finally ready to prove Theorem~\ref{edgepart}.

\begin{proofof}{Theorem~\ref{edgepart}} Let an {\sc 1-in-3 Sat} instance $\phi$ be given with clauses $C_1,\ldots,C_m$ and variables $x_1,\ldots,x_n$. We wish to create a triangle-free graph $G_{\phi}$ that can be edge-partitioned into 4-stars if and only if $\phi$ is 1-in-3 satisfiable.

$G_\phi$ is constructed as follows:\begin{itemize} \item For each variable $x_i$, let $k_i$ denote the number of clauses that $x_i$ occurs in (or the number of clauses that $\bar{x_i}$ occurs in, whichever is greater). Let $d_i$ be the integer satisfying $3 \cdot 2^{d_i-2} < 3(k_i+1) \leq 3 \cdot 2^{d_i-1}$. Add a copy of the graph $G_{d_i}$ to $G_{\phi}$, calling it $A_i$. Note that $A_i$ has at least $3(k_i + 1)$ leaves. \item For each clause $C_i = (l_1 \vee l_2 \vee l_3)$, add three copies of $S_5$ to $G_{\phi}$, calling them $B_{i,1},B_{i,2},B_{i,3}$. \end{itemize} Join the shared edges of these subgraphs as follows: if the literal $l_j = x_k$ is in $C_i$, then merge one shared edge from each of $B_{i,1},B_{i,2},B_{i,3}$ with three shared edges from the top tree of $A_i$; otherwise, if $l_j = \bar{x_k}$ is in $C_i$, then merge a shared edge from each of $B_{i,1},B_{i,2},B_{i,3}$ with three shared edges from the bottom tree of $A_i$.  As a heuristic use a shared edges which is incident to another unused shared edge in $G_{d_k}$, whenever possible.

Since $A_i$ has a $3 \cdot 2^{d_i-1} \geq 3(k_i+1)$ leaves, there may remain some shared edges in some $A_i$ that have not been merged with shared edges from copy of $S_5$. We deal with these extra shared edges as follows: Take three shared edges from different parents in the top tree of $A_i$, and merge their end vertices with a new vertex to form a $4$-star.  Repeat until all unused shared edges from the top are used and do the same for the shared edges on the bottom.  Note that this is possible because we used three copies of $S_5$ for each clause; hence, shared edges from the top/bottom of each gadget are taken in multiples of three, and the number of leaves in every $A_i$ is a multiple of three.  By the above heuristic for choosing unused shared edges we will never create a multi edge.

Clearly the above reduction can be done in polynomial time.  Also note that by construction, $G_\phi$ contains no triangles. We now argue that the formula $\phi$ is 1-in-3 satisfiable if and only if $G_\phi$ can be edge-partitioned into $4$-stars. Supposing that $\phi$ is satisfiable, partition each variable gadget according to its assignment in a given satisfying assignment. In particular, if a variable is set to true, then {\em true-partition} the edges in its corresponding variable gadget. Each clause gadget can be partitioned into a $4$-star, since exactly one of its shared edges are used.  The remaining edges are already part of $4$-stars by construction and can hence be partitioned.

For the other direction, suppose that $G_\phi$ can be partitioned into 4-stars. By Claim~\ref{truefalse}, each copy of $G_d$ is either true or false partitioned.  Now each clause gadget can be partitioned if and only if exactly one of its shared edges is used by a 4-star with a center in a variable gadget. By construction, this happens iff exactly one of the literals in the clause was assigned true in the partition for its variable gadget. Thus the partition defines a satisfying assignment for $\phi$.

Finally, note that the $G_{\phi}$ constructed in Theorem~\ref{edgepart} is triangle-free; in particular, $G_{\phi}$ is bipartite. To see this, note that each $A_i$ is bipartite, each $B_{i,j}$ is bipartite, and for each of these subgraphs, its set of shared edges come from only one side of its bipartition. 
\end{proofof}

While we have given a complete description above, the construction of $G_{\phi}$ is perhaps better understood through examples. We have provided two examples in  Appendix~\ref{edgePartInto4StarsExamples}.

\begin{corollary}
{\sc Edge Partition Into Triangles and 4-Stars} is {\sf NP}-Complete.
\end{corollary}

\begin{corollary}
{\sc Binary 3-Anonymity} is {\sf NP}-Complete.
\end{corollary}

\begin{proof}
Aggarwal {\em et al.}\cite{aggarwal2005at} showed that there is a polynomial time reduction from {\sc Edge Partition Into Triangles And 4-Stars} to {\sc Binary 3-Anonymity}.  Their reduction is repeated in Appendix~\ref{edgePartIntoTrianglesAnd4StarsToBinary3Anonymity} for completeness.  \QED
\end{proof}

\section{Hardness of Computing $\ell$-diversity}

\label{diversity}

Finally, we consider an alternative privacy model called $\ell$-diversity, which strengthens the privacy guarantees of the $k$-anonymity model. It was first proposed to prevent certain background knowledge attacks which could potentially be used against a $k$-anonymized dataset~\cite{machanavajjhala:dpb}. In the model, we distinguish between which attributes of the database are merely potentially identifying and which are highly sensitive. Those which are highly sensitive require a strong privacy guarantee.

\begin{definition}
The \emph{cost} of a $\ell$-diverse solutions is the number of stars introduced, among the attributes $q \in Q$, to the database.
\end{definition}

The fact that optimal $3$-diversity with binary attributes and one sensitive ternary attribute is {\sf NP}-hard should not be too surprising, in light of our proofs of hardness for $3$-anonymity.  Intuitively, the extra sensitive attribute constraint should make $3$-diversity only harder than $3$-anonymity.  What is perhaps surprising is that optimal $2$-diversity is {\sf NP}-hard for databases with three sensitive attributes per row, in light of our result that optimal $2$-anonymity is in $P$.

\begin{reminder}{\ref{2div}}
Optimal $2$-diversity with binary attributes and three sensitive attributes is {\sf NP}-hard.
\end{reminder}

\begin{proof}
The reduction is from edge partition into triangles which is known to be NP-Complete even when the graph is tripartite ~\cite{garey1979cai}.  The idea for the reduction is similar to the reductions in \cite{aggarwal2005at} for binary $k$-anonymity (see Appendix~\ref{AggarwalReductionLemma}). Given a graph $G=(V,E)$, define a $2$-diversity instance as follows:  the rows of the table correspond to each $e \in E$, while the columns correspond to the $n = |V|$ vertices of $G$ plus the sensitive attributes $s_{i_0},s_{i_1},s_{i_2}$.  Given an arbitrary ordering of the vertices $V = \{v_1,...,v_n\}$ and edges $E = \{e_1,...,e_m\}$ define a matrix $R^G$ as follows:

\[ \mbox{$R^G[i][j]=$} \left\{ \begin{array}{ll}
         1 & \mbox{if $v_j \in e_i$};\\
         0 & otherwise.\end{array} \right. \]

Let $V_0,V_1,V_2$ be the tripartition of vertices in G.  Now label:

\[ \mbox{$s_{i_j}=$} \left\{ \begin{array}{ll}
         0 & \mbox{if $e_i \bigcap V_j = \emptyset$};\\
			1 & \mbox{otherwise};
\end{array}          \right. \]

The cost of grouping any three rows in a $2$-diverse solution is at least three stars because the graph is simple. Furthermore, any group of more than three rows will require more than three stars per row to $2$-diversify.  The proof is argument is identical to lemma \ref{AggarwalReductionLemma} in Appendix~\ref{edgePartIntoTrianglesAnd4StarsToBinary3Anonymity}. \newline \newline

\begin{lemma}
Any group of only two rows in $R^G$ violates the $2$-diversity constraint.
\end{lemma}
\begin{proof}
Let $i,j$ be any pair of distinct rows. Because the graph is tripartite, either $s_{i_1} = s_{j_1}$ or $s_{i_2} = s_{j_2}$ or else $s_{i_3} = s_{j_3}$. The diversity constraints for two rows will look like:

\begin{center}
\begin{tabular}{|c|c|c|c|c|}
\hline
 & $Q$ & \multicolumn{3}{c|}{$S$} \\
	\hline
$e_i$ & $\ldots$ &  $1$& $1$& $0$    \\
$e_j$ & $\ldots$ & $1$ &$0$& $1$    \\
	\hline
\end{tabular}\end{center} \QED \end{proof}

Similarly, the diversity constraints coupled with the fact that the graph is 3-Partite also prevent us from choosing three rows corresponding to a $4$-star in G because the rows would share a sensitive attribute.  However, the diversity constraints do allow for the possibility that the three rows correspond to a triangle in G as illustrated in the table:

\begin{center}
\begin{tabular}{|c|c|c|c|c|}
\hline
 & $Q$ & \multicolumn{3}{c|}{$S$} \\
	\hline
$e_i$ & & $1$ &$1$ & $0$    \\
$e_j$ &  &$1$ & $0$ &$1$    \\
$e_k$ &  & $0$& $1$ & $1$    \\
	\hline
\end{tabular}\end{center}

Thus the edges of G can be partitioned into triangles iff the $2$-diversity instance has a solution that introduces exactly 3 stars per row. \QED
\end{proof}

\begin{reminder}{\ref{3div}}
Optimal $3$-diversity with binary attributes is {\sf NP} hard, with only one sensitive ternary attribute.
\end{reminder}

\begin{proof}
The hardness reduction for $3$-diversity with one sensitive attribute is essentially the same as above. Assume that $G$ is tripartite, and let $V_1,V_2,V_3$ be the three partite sets in $G$. Let $s_i$ denote the sensitive attribute for row $v_i$.  If $3$-diversity is to be feasible then the sensitive attribute $s_i$ must be allowed to take at least three values.  Other attributes must be binary.

\[ \mbox{$s_{i}=$} \left\{ \begin{array}{ll}
         1 & \mbox{if $e_i = (x,y),$ with $x \in V_1, y \in V_2$};\\
			2 & \mbox{if $e_i = (x,z),$ with $x \in V_1, z \in V_3$};\\
			3 & \mbox{if $e_i = (y,z),$ with $y \in V_2, z \in V_3 $};
\end{array}          \right. \]

As before, the diversity constraints now prevent us from grouping three rows which correspond to a $4$-star.  Groups of rows which do not correspond to a triangle in G still require more than three stars per row.  Thus the edges of G can be partitioned into triangles iff the $3$-diversity instance has a solution that introduces exactly 3 stars per row. \QED
\end{proof}

\section{Conclusion}

We have demonstrated the hardness and feasibility of several methods used in database privacy, settling several open problems on the topic. The upshot is that most of these problems are difficult to solve optimally, even in very special cases; however in some interesting cases these problems can be solved faster. Several interesting open questions address possible ways around this intractability:

\begin{itemize}

\item To what degree can the hard problems be approximately solved? For example, the best known approximation algorithm for $k$-anonymity, given by Park and Shim~\cite{park2007aak}, suppresses no more than $O(\log k)$ times the optimal number of entries. Could better approximation ratios be achieved when the number of attributes is small?


\item The best known running time for Simplex Matching is $O(n^3+n^2m^2)$ steps ~\cite{anshelevich2007tbm}.  Here, $n$ is the number of nodes and $m$ is the number of hyperedges in the hypergraph.  In our algorithm for $2$-anonymity, $n$ is also the number of rows in the database while $m = \binom{n}{3} = O(n^3)$ because we add a hyperedge for every triples. Hence our algorithm for $2$-Anonymity has running time $O(n^8)$. Can this exponent be reduced to a more practical running time?

\end{itemize}

\acks{We would like to thank Manuel Blum, Lenore Blum and Anupam Gupta for their help and guidance during this work.}

\vskip 0.2in
\bibliography{kAnonymityPaper}
\appendix
\section{Proof of Lemma \ref{kAnonConstAttribReductionLemma}}
\label{kAnonConstAttribReductionLemmaProof}

Recall we had the following three properties of the database $D$ in the reduction of Theorem~\ref{3anon27}: 

\begin{enumerate}
\item For any $x$ rows  $r_i,r_j,r_k,r_l$, where $x \geq 4$, the cost of anonymizing these rows is
$$C_{i,j,k,l} = x \cdot 27$$ because there is no alphabet symbol that is used in all $4$ rows.
\item If \mbox{$\{r_i,r_j,r_k\} \notin M$} then the cost of anonymizing the three corresponding rows is
$$C_{i,j,k} = 3\cdot 27$$ because there is no alphabet symbol that is used in all $3$ rows.
\item If \mbox{$\{r_i,r_j,r_k\} \in M$} then the cost of anonymizing the three corresponding rows is
$$C_{i,j,k} = 3\cdot 26$$
because the three rows will match in exactly one of the 27 columns.
\end{enumerate}

\begin{reminderlemma}{\ref{kAnonConstAttribReductionLemma}}
There is a polynomial time mapping $g$ from 3DM-3 feasible solutions to 3-anonymity feasible solutions, such that if $M' \subseteq M$ is a 3DM-3 feasible solution then $C_{3DM}(M') = 27 C_{3ANON}(g(M'))$.
\end{reminderlemma}

\begin{proof} We use the reduction defined in the proof of Theorem~\ref{3anon27}.
Recall that in a $3$-anonymity solution $P$ is a partition of the rows into groups of size $3, 4$ and $5$.  By the three properties of $D$, any group which does not correspond to triple from $M$ must be suppressed entirely.  Hence, we can think of the solution as a partition of the rows into triples $(x_i,y_j,w_k)$ from $M$ and some other rows.  Similarly, we can think of a 3DM solution as a partition of the elements into triples from $M$ and some other elements.  Thus we can define a polynomial time computable transformation $f$ between 3DM-3 solutions and 3-anonymity solutions. 

By the above properties of $D$, $Cost^{\star}_3(g(M')) = 27n - 3|M'|$.  Therefore,
\begin{eqnarray*}
C_{3DM}(M') &=& \frac{3|M'|}{|X|+|Y|+|W|} \\
&=& \frac{3|M'|}{ n} \\
&=& 27-\frac{27n -3|M'|}{ n} \\
&=& 27-\frac{Cost^{\star}_3(g(M')}{ n} \\
&=& 27\cdot C_{3ANON}(g(M'))
\end{eqnarray*}
\QED
\end{proof}
\section{Edge Partition into 4-Star Reduction - Examples}

\label{edgePartInto4StarsExamples}

Figure~\ref{reductionGadgets} shows examples of a variable gadget and a clause gadget.

\begin{figure}[h]
\begin{center}
\subfigure[Example gadget: $G_3$.  The private edges are solid and the shared edges are dashed.  The selection of the three deleted leaves is arbitrary.]{\label{variableGadget}
\begin{picture}(78,94)(0,-94)
\put(0,-94){\framebox(78,94){}}
\gasset{Nw=4.0,Nh=4.0,Nmr=2.0,Nfill=y,fillcolor=Black}
\gasset{AHnb=0}
\node(n73)(2.82,-31.88){}

\node(n74)(4.0,-60.0){}

\node(n75)(12.0,-16.0){}

\node(n76)(12.0,-76.0){}

\node(n77)(31.97,-16.0){}

\node(n78)(63.82,-16.0){}

\node(n79)(31.97,-4.0){}

\node(n81)(32.03,-76.0){}

\node(n82)(64.0,-76.0){}

\node(n83)(31.97,-88.0){}

\node(n84)(20.0,-60.0){}

\node(n85)(28.0,-60.0){}

\node(n86)(44.0,-60.0){}

\node(n87)(59.85,-60.0){}

\node(n88)(67.85,-60.0){}

\node(n89)(59.88,-32.0){}

\node(n90)(67.88,-32.0){}

\node(n91)(44.0,-32.0){}

\node(n92)(28.0,-32.0){}

\node(n94)(8.0,-52.0){}

\node(n95)(8.0,-40.0){}

\node(n96)(15.97,-52.0){}

\node(n97)(20.0,-32.0){}

\node(n98)(15.97,-40.0){}

\drawedge(n79,n77){}

\drawedge(n78,n79){}

\drawedge(n79,n75){}

\drawedge(n75,n73){}

\drawedge[dash={2.0 2.0 2.0 3.0}{0.0}](n73,n95){ }

\drawedge(n73,n74){}

\drawedge[dash={2.0 2.0 2.0 3.0}{0.0}](n74,n94){ }

\drawedge(n74,n76){}

\drawedge(n76,n83){}

\drawedge(n83,n81){}

\drawedge(n83,n82){}

\drawedge(n85,n81){}

\drawedge(n97,n84){}

\drawedge[dash={2.0 2.0 2.0 3.0}{0.0}](n84,n96){ }

\drawedge[dash={2.0 2.0 2.0 3.0}{0.0}](n98,n97){ }

\drawedge(n92,n77){}

\drawedge(n75,n97){}

\drawedge(n77,n91){}

\drawedge(n76,n84){}

\drawedge(n86,n81){}

\drawedge(n87,n82){}

\drawedge(n88,n82){}

\drawedge(n78,n89){}

\drawedge(n90,n78){}

\drawedge(n90,n88){}

\node(n100)(71.88,-52.0){}

\node(n101)(71.88,-40.0){}

\node(n102)(63.88,-40.0){}

\node(n103)(55.88,-40.0){}

\node(n104)(55.88,-52.0){}

\node(n105)(63.88,-52.0){}

\node(n106)(24.03,-52.0){}

\drawedge[dash={2.0 2.0 2.0 3.0}{0.0}](n100,n88){ }

\drawedge[dash={2.0 2.0 2.0 3.0}{0.0}](n101,n90){ }

\drawedge[dash={2.0 2.0 2.0 3.0}{0.0}](n102,n89){ }

\drawedge[dash={2.0 2.0 2.0 3.0}{0.0}](n89,n103){ }

\drawedge(n104,n87){}

\drawedge[dash={2.0 2.0 2.0 3.0}{0.0}](n87,n105){ }

\node(n114)(32.03,-52.0){}

\node(n115)(48.0,-52.0){}

\node(n116)(40.0,-52.0){}

\node(n117)(24.03,-40.1){}

\node(n118)(32.03,-40.1){}

\node(n119)(48.0,-40.1){}

\node(n120)(40.0,-40.0){}

\drawedge[dash={2.0 2.0 2.0 3.0}{0.0}](n86,n115){ }

\drawedge[dash={2.0 2.0 2.0 3.0}{0.0}](n116,n86){ }

\drawedge[dash={2.0 2.0 2.0 3.0}{0.0}](n119,n91){ }

\drawedge[dash={2.0 2.0 2.0 3.0}{0.0}](n91,n120){ }

\drawedge[dash={2.0 2.0 2.0 3.0}{0.0}](n118,n92){ }

\drawedge[dash={2.0 2.0 2.0 3.0}{0.0}](n92,n117){ }

\drawedge[dash={2.0 2.0 2.0 3.0}{0.0}](n106,n85){ }

\drawedge[dash={2.0 2.0 2.0 3.0}{0.0}](n114,n85){ }

\end{picture}
}
\subfigure[$S_5$: 4-Star Clause Gadget]{\label{4StarClauseGadget}
\begin{picture}(44,24)(0,-24)
\put(0,-24){\framebox(38,24){}}
\gasset{Nw=4.0,Nh=4.0,Nmr=2.0,Nfill=y,fillcolor=Black,AHnb=0}
\node[NLangle=352.97,NLdist=41.05,Nadjustdist=2.9](n66)(20.0,-4.0){ }

\node(n67)(8.0,-4.0){}

\node(n68)(8.0,-16.0){}

\node(n69)(20.0,-16.0){}

\node(n70)(32.0,-16.0){}

\drawedge[dash={2.0 2.0 2.0 3.0}{0.0}](n66,n68){ }

\drawedge[dash={2.0 2.0 2.0 3.0}{0.0}](n66,n69){ }

\drawedge[dash={2.0 2.0 2.0 3.0}{0.0}](n66,n70){ }

\gasset{dash={2.0 2.0 2.0 3.0}{0.0}}
\drawedge[dash={}{0.0}](n67,n66){ }

\end{picture}
}
\caption{Gadgets}
\label{reductionGadgets}

\end{center}

\end{figure}

\paragraph{Example 1.} $ \phi = ({\bf \bar{x}} , {\bf y} , {\bf z}) ({\bf x} , {\bf \bar{ y}} , {\bf z})$. Note that the {\sc 1-in-3 Sat} formula has two satisfying assignments: $(x=t, y = t, z = f)$, $(x=f,y=f,z=f)$. Similarly, the corresponding graph $G_\phi$ (shown in figure~\ref{4StarReduction1}) can be partitioned into 4-stars in exactly two ways, both corresponding to the satisfying assignments.

\begin{figure}[h]
\begin{center}
\begin{picture}(133,97)(0,-97)
\put(0,-97){\framebox(133,97){}}
\gasset{Nw=3.0,Nh=3.0,Nmr=1.5,Nfill=y,fillcolor=Black,AHnb=0}
\node[ExtNL=y,NLangle=0.0,NLdist=1.4](n0)(27.91,-31.87){X}

\node(n1)(19.91,-39.87){}

\node(n2)(27.91,-39.87){}

\node(n3)(35.91,-39.87){}

\node(n4)(19.91,-59.87){}

\node(n5)(27.91,-59.87){}

\node(n6)(35.91,-59.87){}

\node[ExtNL=y,NLangle=0.0,NLdist=1.6](n7)(27.91,-67.87){$\bar{X}$}

\node[ExtNL=y,NLangle=0.0,NLdist=1.5](n8)(108.02,-23.9){Z}

\node(n9)(92.05,-31.87){}

\node(n10)(107.99,-31.9){}

\node(n11)(124.02,-31.9){}

\node(n12)(96.02,-39.9){}

\node(n13)(88.02,-39.9){}

\node(n14)(104.02,-39.9){}

\node(n15)(111.99,-39.9){}

\node(n16)(119.96,-39.9){}

\node(n17)(128.02,-39.9){}

\drawedge(n8,n9){}

\drawedge(n8,n10){}

\drawedge(n8,n11){}

\drawedge(n16,n11){}

\drawedge(n17,n11){}

\drawedge(n10,n15){}

\drawedge(n10,n14){}

\drawedge(n9,n12){}

\drawedge(n9,n13){}

\drawedge(n7,n6){}

\drawedge(n5,n7){}

\drawedge(n4,n7){}

\drawedge(n1,n4){}

\drawedge(n2,n5){}

\drawedge(n3,n6){}

\drawedge(n0,n2){}

\drawedge(n3,n0){}

\drawedge(n1,n0){}

\node[ExtNL=y,NLangle=0.0,NLdist=1.5](n31)(108.02,-79.84){$\bar{Z}$}

\node(n32)(92.05,-71.82){}

\node(n33)(107.99,-71.85){}

\node(n34)(124.02,-71.85){}

\node(n35)(96.02,-63.84){}

\node(n36)(88.02,-63.98){}

\node(n37)(103.99,-63.98){}

\node(n38)(111.99,-63.98){}

\node(n39)(119.96,-63.98){}

\node(n40)(128.02,-63.98){}

\drawedge(n31,n32){}

\drawedge(n31,n33){}

\drawedge(n31,n34){}

\drawedge(n39,n34){}

\drawedge(n40,n34){}

\drawedge(n33,n38){}

\drawedge(n33,n37){}

\drawedge(n32,n35){}

\drawedge(n32,n36){}

\drawedge(n34,n31){}

\drawedge(n13,n36){}

\drawedge(n37,n14){}

\drawedge(n12,n35){}

\node(n51)(92.02,-55.93){}

\node(n52)(92.05,-47.9){}

\node(n53)(116.21,-51.8){}

\node(n54)(124.02,-51.93){}

\drawedge(n13,n52){}

\drawedge(n52,n14){}

\drawedge(n15,n52){}

\drawedge(n54,n38){}

\drawedge(n54,n39){}

\drawedge(n54,n40){}

\drawedge(n40,n53){}

\drawedge(n39,n53){}

\drawedge(n38,n53){}

\drawedge(n51,n37){}

\drawedge(n51,n35){}

\drawedge(n51,n36){}

\node[ExtNL=y,NLangle=0.0,NLdist=1.4](n79)(67.86,-31.87){Y}

\node(n80)(59.86,-39.87){}

\node(n81)(67.86,-39.87){}

\node(n82)(75.86,-39.87){}

\node(n83)(59.86,-59.87){}

\node(n84)(67.97,-60.0){}

\node(n85)(75.86,-59.87){}

\node[ExtNL=y,NLangle=0.0,NLdist=1.6](n86)(67.86,-67.87){$\bar{Y}$}

\drawedge(n86,n85){}

\drawedge(n84,n86){}

\drawedge(n83,n86){}

\drawedge(n80,n83){}

\drawedge(n81,n84){}

\drawedge(n82,n85){}

\drawedge(n79,n81){}

\drawedge(n82,n79){}

\drawedge(n80,n79){}

\node[ExtNL=y,NLangle=0.0,NLdist=1.5](n95)(56.06,-83.94){Clause 1 ($\bar{X},Y,Z$)}

\node(n96)(56.06,-87.94){}

\node(n97)(56.06,-91.94){}

\node(n98)(44.06,-83.94){}

\node(n99)(44.06,-87.94){}

\node(n100)(44.06,-91.94){}

\drawedge(n100,n97){}

\drawedge(n99,n96){}

\drawedge(n98,n95){}

\node[ExtNL=y,NLangle=358.53,NLdist=8.82](n101)(56.0,-4.0){Clause 2 (X,$\bar{Y}$,Z)}

\node(n102)(55.94,-8.0){}

\node(n103)(56.07,-12.28){}

\node(n104)(44.07,-4.29){}

\node(n105)(44.0,-8.0){}

\node(n106)(44.07,-12.28){}

\drawedge(n106,n103){}

\drawedge(n105,n102){}

\drawedge(n104,n101){}

\drawedge(n95,n4){}

\drawedge(n96,n5){}

\drawedge(n97,n6){}

\drawedge(n95,n12){}

\drawedge(n96,n15){}

\drawedge(n97,n16){}

\drawedge(n103,n16){}

\drawedge(n102,n17){}

\drawedge(n17,n101){}

\drawedge(n95,n83){}

\drawedge(n96,n84){}

\drawedge(n97,n85){}

\drawedge(n1,n103){}

\drawedge(n2,n102){}

\drawedge(n3,n101){}

\drawedge(n80,n103){}

\drawedge(n82,n101){}

\drawedge(n102,n81){}
\end{picture}
\caption{\small  $\phi = ({\bf \bar{x}} , {\bf y} , {\bf z}) ({\bf x} , {\bf \bar{ y}} , {\bf z})$}
\label{4StarReduction1}
\end{center}
\end{figure}
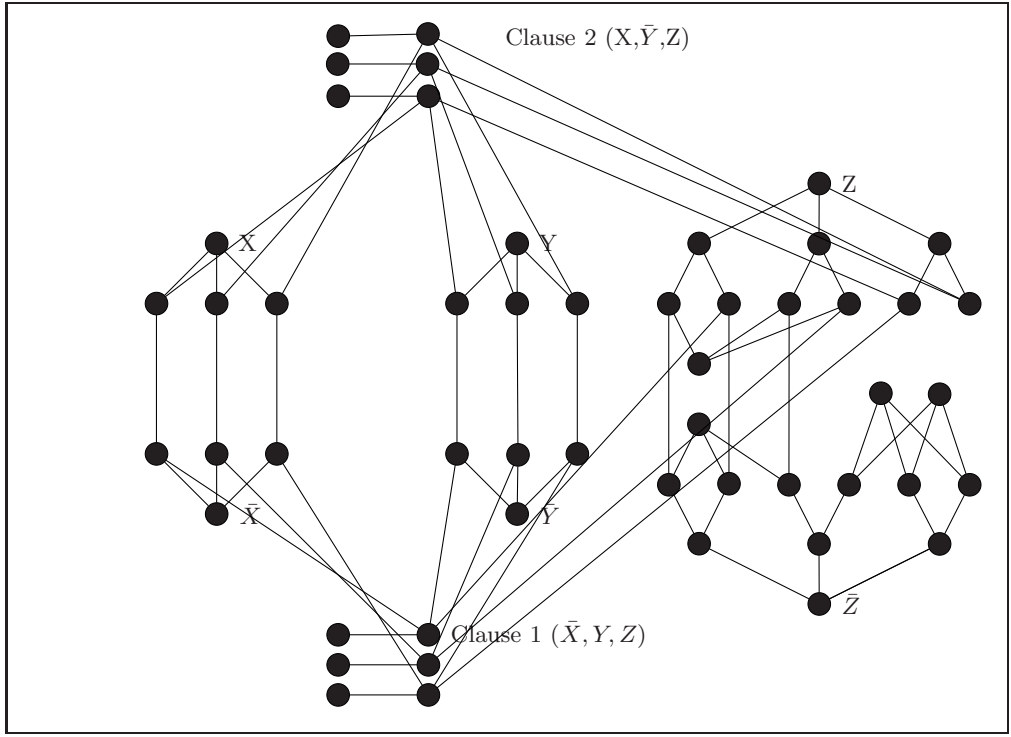

\paragraph{Example 2.} $\phi = ({\bf x} , {\bf y} , {\bf z})  ({\bf \bar{x}} , {\bf \bar{ y}} , {\bf \bar{z}})$. Note that $\phi$ is {\em not} 1-in-3 satisfiable. Similarly, the corresponding graph $G_\phi$ (shown in figure~\ref{4StarReduction2}) cannot be edge-partitioned into 4-stars.

\begin{figure}[h]

\begin{center}
\begin{picture}(109,101)(0,-101)
\put(0,-101){\framebox(109,101){}}
\gasset{Nw=3.0,Nh=3.0,Nmr=1.5,Nfill=y,fillcolor=Black,AHnb=0}
\node[ExtNL=y,NLangle=0.0,NLdist=1.1](n1)(20.0,-28.0){X}

\node(n2)(12.0,-36.0){}

\node(n3)(19.97,-35.84){}

\node(n4)(27.97,-35.84){}

\node(n5)(11.97,-51.84){}

\node(n6)(19.97,-51.84){}

\node(n7)(28.0,-52.0){}

\node[NLangle=0.0,NLdist=7.0](n8)(19.97,-63.84){$\bar{X}$}

\drawedge(n1,n3){}

\drawedge(n4,n1){}

\drawedge(n1,n2){}

\drawedge(n5,n8){}

\drawedge(n6,n8){}

\drawedge(n7,n8){}

\drawedge(n4,n7){}

\drawedge(n3,n6){}

\drawedge(n2,n5){}

\node[ExtNL=y,NLangle=0.0,NLdist=1.1](n9)(56.22,-27.84){Y}

\node(n10)(48.22,-35.84){}

\node(n11)(56.0,-36.0){}

\node(n12)(64.22,-35.84){}

\node(n13)(48.22,-51.84){}

\node(n14)(56.22,-51.84){}

\node(n15)(64.22,-51.84){}

\node[ExtNL=y,NLangle=0.0,NLdist=1.0](n16)(56.22,-63.84){$\bar{ Y}$}

\drawedge(n9,n11){}

\drawedge(n12,n9){}

\drawedge(n9,n10){}

\drawedge(n13,n16){}

\drawedge(n14,n16){}

\drawedge(n15,n16){}

\drawedge(n12,n15){}

\drawedge(n11,n14){}

\drawedge(n10,n13){}

\node[ExtNL=y,NLangle=0.0,NLdist=1.0](n25)(95.97,-28.0){Z}

\node(n26)(87.91,-36.1){}

\node(n27)(95.91,-36.1){}

\node(n28)(103.91,-36.1){}

\node(n29)(87.91,-52.1){}

\node(n30)(95.91,-52.1){}

\node(n31)(104.0,-52.0){}

\node[ExtNL=y,NLangle=0.0,NLdist=1.0](n32)(95.91,-64.1){$\bar{ Z}$}

\drawedge(n25,n27){}

\drawedge(n28,n25){}

\drawedge(n25,n26){}

\drawedge(n29,n32){}

\drawedge(n30,n32){}

\drawedge(n31,n32){}

\drawedge(n28,n31){}

\drawedge(n27,n30){}

\drawedge(n26,n29){}

\node(n41)(46.0,-4.0){}

\node(n42)(45.94,-8.0){}

\node[ExtNL=y,NLangle=0.0,NLdist=1.0](n43)(45.97,-12.0){Clause 1 (X,Y,Z)}

\node(n44)(30.06,-4.0){}

\node(n45)(30.09,-8.0){}

\node(n46)(30.06,-12.0){}

\drawedge(n44,n41){}

\drawedge(n45,n42){}

\drawedge(n43,n46){}

\drawedge(n2,n43){}

\drawedge(n3,n42){}

\drawedge(n4,n41){}

\drawedge(n10,n43){}

\drawedge(n11,n42){}

\drawedge(n12,n41){}

\drawedge(n28,n41){}

\drawedge(n42,n27){}

\drawedge(n43,n26){}

\node(n47)(45.98,-88.0){}

\node(n48)(45.94,-91.87){}

\node[ExtNL=y,NLangle=0.0,NLdist=1.0](n49)(46.0,-96.12){Clause 2 ($\bar{X},\bar{Y},\bar{Z}$)}

\node(n50)(30.32,-87.87){}

\node(n51)(30.09,-91.87){}

\node(n52)(30.06,-95.86){}

\drawedge(n50,n47){}

\drawedge(n51,n48){}

\drawedge(n49,n52){}

\drawedge(n5,n49){}

\drawedge(n13,n49){}

\drawedge(n29,n49){}

\drawedge(n48,n6){}

\drawedge(n48,n14){}

\drawedge(n48,n30){}

\drawedge(n47,n7){}

\drawedge(n47,n15){}

\drawedge(n47,n31){}

\end{picture}
\caption{\small $\phi = ({\bf x} , {\bf y} , {\bf z})  ({\bf \bar{x}} , {\bf \bar{ y}} , {\bf \bar{z}})$}
\label{4StarReduction2}
\end{center}
\end{figure}

\section{Reducing Edge Partition Into Triangles And 4-Stars to Binary 3-Anonymity}
\label{edgePartIntoTrianglesAnd4StarsToBinary3Anonymity}
Given a graph $G = (V,E)$ with $m$ edges and $n$ vertices build the following table: the rows of the table correspond to each edge $e \in E$, while the columns correspond to the $n = |V|$ vertices of $G$.  Given an arbitrary ordering of the vertices $V = \{v_1,...,v_n\}$ and edges $E = \{e_1,...e_m\}$ define a database $R^G$ as follows:

\[ \mbox{$R^G[i][j]=$} \left\{ \begin{array}{ll}
         1 & \mbox{if $v_j \in e_i$};\\
         0 & otherwise.\end{array} \right. \]

Clearly this reduction takes polynomial time. Note that, because the graph $G$ is simple, any $3$-anonymous solution must include at least three stars per row. This follows because for any set of three edges, there are at least three vertices that are incident with one, but not all, of the three edges.  Furthermore, note that if a set of three edges do not form a triangle or $4$-star, then there are at least four vertices that are incident with one (but not all) of the three edges.  The result follows from lemma \ref{AggarwalReductionLemma}.

\begin{lemma}
\label{AggarwalReductionLemma}
Let $m$ be the number of edges in $G$, the cost of the optimal $3$-anonymous solution for $R^G$ is $3m$ stars iff the graph $G$ can be edge partitioned into $4$-stars and triangles.
\end{lemma}
\begin{proof}
First, suppose that the cost of the optimal {\sc $3$-Anonymity} solution is $3m$. Since each row has at least 3 stars in it, each row must have exactly three stars because there are $m$ rows in $R^G$.  Given a set of three identical (anonymized) rows, each row has 3 stars each corresponding to a vertex that was incident to one, but not all of the three edges represented by those rows.  Hence, those three edges form either a 4-star or a triangle.  Therefore, the edges of the graph can be partitioned into triangles and 4-stars.

For the other direction, suppose that $G$ can be partitioned into triangles and 4-stars.  Group the rows of the 3-anonymity instance according to the partition. Now consider a group of three rows in the table that correspond to the edges of a 4-star. The three rows have the form:
\begin{center}
\begin{tabular}{|c|c|}
	\hline
$(v_0,v_1)$ &  $\cdots 1100 \cdots$ \\
$(v_0,v_2)$ &  $\cdots 1010 \cdots$  \\
$(v_0,v_3)$ &  $\cdots 1001 \cdots$   \\
	\hline
\end{tabular}
\end{center}

where the $\cdots$ are all $0$'s. For a triangle, the three rows corresponding to its edges looks like:
\begin{center}
\begin{tabular}{|c|c|}
	\hline
$(v_0,v_1)$ &  $\cdots 110 \cdots$    \\
$(v_0,v_2)$ &  $\cdots 101 \cdots$    \\
$(v_1,v_2)$ &  $\cdots 011 \cdots$    \\
	\hline
\end{tabular}\end{center}

where again the $\cdots$ are all $0$'s. Clearly, both groups of rows can be made identical by suppressing only three entries per row.  Hence, the table can be made 3-anonymous with $3m$ stars. \QED
\end{proof}

\section{Edge Partitioning $G_d$ into 4-Stars}
\label{proofOfGadgetClaim}
Recall the statement of lemma \ref{truefalse} 
\begin{lemma} Let $G$ be a graph containing a share-respecting copy of $G_d$. Assuming there is an edge partition of $G$ into $4$-stars, exactly one of two cases must hold for that partition:  
\begin{enumerate}
\item All shared edges belonging to the top tree of $G_d$ are contained in $4$-stars with centers in $G_d$, while all shared edges belonging to the bottom tree are contained in 4-stars with centers not contained in $G_d$.
\item All shared edges belonging to the bottom tree of $G_d$ are contained in 4-stars with centers in $G_d$, while all shared edges belonging to the bottom tree are contained in 4-stars with centers not contained in $G_d$.
\end{enumerate}
\end{lemma}
\begin{proof}
Let $G$ be a graph which contains a share-respecting copy of $G_d$ along with an edge partition $P$ of $G$ into $4$-stars. Note that every vertex in $G_d$ has degree $3$, even the leaves in $G_d$ are have two more shared edges in $G_d$. If an internal vertex $x$ in a 3-Binary Tree is the center of a $4$-star in $P$ (see figure~\ref{3BinTreeProofExample}) then its parent $y$ cannot be the center of any $4$-star in $P$ because its degree has been reduced to $2$. This means that $z$ must be the center of a $4$-star in $P$ to cover the edge $(y,z)$.  Similarly, if $x$ was not the center of a 4-star then $y$ must be the center of a 4-star to cover the edge $(x,y)$. 

\begin{figure}[h]

\begin{center}

 \label{3BinTreeProofExample}
\begin{picture}(43,46)(0,-46)
\put(0,-46){\framebox(43,46){}}
\gasset{Nw=4.0,Nh=4.0,Nmr=2.0,Nfill=y,fillcolor=Black}
\node[fillcolor=Cyan,NLangle=0.0,Nw=5.3,Nh=5.3,Nmr=2.65](n0)(20.0,-28.0){x}

\node(n1)(12.0,-40.0){}

\node(n2)(24.0,-40.0){}

\node[NLangle=0.0,NLdist=3.8](n3)(28.0,-16.0){y}

\node[fillcolor=Blue,NLangle=0.0,Nw=5.2,Nh=5.3,Nmr=2.6](n4)(36.0,-4.0){z}

\node[fillgray=1.0](n5)(32.0,-28.0){}

\gasset{AHnb=0}
\drawedge(n4,n3){}

\drawedge(n3,n5){}

\drawedge[linecolor=Cyan,linewidth=0.23](n3,n0){ }

\drawedge[linecolor=Cyan,linewidth=0.25](n0,n2){ }

\drawedge[linecolor=Cyan,linewidth=0.25](n0,n1){ }
\end{picture}
\caption{$x$ is the center of a $4$-star, therefore $y$ cannot be the center of a $4$-star.  In any partition the edge $(y,z)$ must be covered by a $4$-star centered at $z$.}

\end{center}

\end{figure}
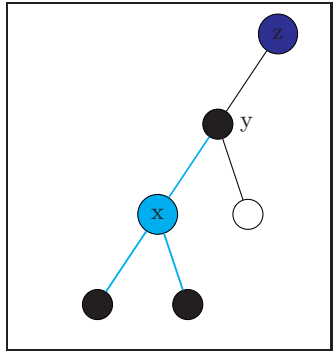

Now the pattern becomes evident: the parent of $z$ cannot be the center of a $4$-star so $z$'s grandparent must be the center of a $4$-star, and so on.
Therefore, in any edge partition, if there is a $4$-star centered at vertex $v$ at depth $i$, then the ancestors of $v$ at depths $i-2,i-4,...$ as well as the descendents at depths $i+2,i+4,\ldots$ must {\em all} be centers of 4-stars as well.  

Consider the root of the $3$-Binary Tree at depth $0$, there are only two possible scenarios.  Scenario 1, the root is the center of a $4$-star and all the vertices (descendents) at depths $0,2,4,\ldots$ in that $3$-Binary Tree must also be the centers of $4$-stars. Scenario 2, the root is not the center of a $4$-star and all the vertices (descendents) at depths $1,3...$ must be centers of $4$-stars.  

By construction of $G_d$ there must be exactly three edges between the top and bottom $3$-Binary Trees in $G_D$. Pick one such edge $(u,v)$, any edge partition of $G$ must use the edge $(u,v)$ so either $u$ or $v$ must be the center of a 4-star.  Without loss of generality assume that $u$ is the center of a $4$-star and that $u$ is in the bottom $3$-Binary Tree.  Notice that both $u$ and $v$ are at depth $d-1$ in their respective trees.  Assume that $d$ is odd (the proof is similar for $d$ even), then in the bottom tree we are in Scenario 1, but in the top tree we are in Scenario 2.  All shared edges belonging to the bottom tree of $G_d$ are contained in $4$-stars centered at depth $d-1$, while no shared edges from the top tree of $G_d$ can be contained in $4$-stars centered at depth $d-1$. \QED
\end{proof}

\end{document}